\documentclass[a4paper,USenglish,numberwithinsect]{lipics-v2016}
\usepackage[T1]{fontenc}

\EventEditors{Rasmus Pagh}
\EventNoEds{1}
\EventLongTitle{15th
Scandinavian Symposium and Workshops on Algorithm Theory (SWAT
2016)}
\EventShortTitle{SWAT 2016}
\EventAcronym{SWAT}
\EventYear{2016}
\EventDate{June 22-24, 2016}
\EventLocation{Reykjavik, Iceland}
\EventLogo{}
\SeriesVolume{53}
\DOI{10.4230/LIPIcs.SWAT.2016}

\graphicspath{{figs/}}
\usepackage{cite} 
\usepackage{microtype} 

\newcommand{\fp}{\phi}

\title{Cuckoo Filter: Simplification and Analysis}

\author{David Eppstein\footnote{Supported in part by NSF grant 1228639. The author would like to thank Michael Mitzenmacher for introducing him to cuckoo filters, and for helpful discussions on this work.}}

\affil{Computer Science Department, University of California, Irvine, USA}

\authorrunning{David Eppstein}
\Copyright{David Eppstein}


\begin{document}
\maketitle

\begin{abstract}
The cuckoo filter data structure of Fan, Andersen, Kaminsky, and Mitzenmacher (CoNEXT 2014)
performs the same approximate set operations as a Bloom filter in less memory, with better locality of reference, and adds the ability to delete elements as well as to insert them. However, until now it has lacked theoretical guarantees on its performance. We describe a simplified version of the cuckoo filter using fewer hash function calls per query. With this simplification, we provide the first theoretical performance guarantees on cuckoo filters, showing that they succeed with high probability whenever their fingerprint length is large enough.
\end{abstract}

\section{Introduction}

Bloom filters~\cite{Blo-CACM-70} are a very widely used data structure for approximately representing sets using low space. At the cost of $O(1)$ bits per element, they can represent any set, with constant-time membership testing, no false negatives, and an arbitrarily low false positive rate controlled by the bits per element. Recently, Fan, Andersen, Kaminsky, and Mitzenmacher~\cite{FanAndKam-CoNEXT-14} proposed an alternative data structure for the same purpose, the \emph{cuckoo filter}. They show experimentally that cuckoo filters are better than Bloom filters in several important ways: they use (up to lower-order terms) 30\% less space for the same false positive rate, matching the information-theoretic lower bound. They have better locality of reference, accessing only two contiguous blocks of memory per query rather than the larger numbers of a typical Bloom filter. And, unlike a Bloom filter, they can handle element deletions as well as insertions and queries without any increase in storage. These good features have already led to the use of cuckoo filters in several applications~\cite{GupBre-ICDF2C-15,GriYavHam-WSCNIS-15}. (For a different and more theoretical replacement for Bloom filters with similar advantages, see Pagh, Pagh, and Rao~\cite{PagPagRao-SODA-05}.)

A cuckoo filter uses a hash table to store a small \emph{fingerprint} for each element, and answers queries by testing whether the fingerprint of the queried element is present. Each element has two hash table cells  where its fingerprints might be stored, determined by a combination of a hash of the element and a second hash of the fingerprint. As in cuckoo hashing~\cite{PagRod-JAlg-04}, fingerprints already stored in the table may be moved to their second location to make room for a newly inserted fingerprint. The performance of a cuckoo filter is controlled by the number $n$ of elements in the set it represents, together with three design parameters: the table size $N$ (number of cells), block size $b$ (fingerprints that can be stored in a single cell), and fingerprint size $f$ (bits per fingerprint). A good choice of these parameters  allows the fingerprints for all elements in the given set to be stored in the table, giving a data structure whose false positive rate $\epsilon$ (the probability that an element not in the set is falsely reported to be in the set) can easily be bounded by $\epsilon\le 2b/(2^f-1)$. For bad choices of parameters, or unlucky choices of hash function, the data structure may fail, being unable to store all its elements' fingerprints. Therefore, it is important to analyze the likelihood of a failure, and to understand which combinations of parameters have a guaranteed low failure probability.

We may define the \emph{load factor} $\frac{n}{bN}$ to be the ratio of the number of stored fingerprints to the number that could be stored. As this number will be close to one, it is convenient to represent its difference from one as a parameter $\delta$, with $1-\delta=\frac{n}{bN}$.
The experiments of Fan et al.~\cite{FanAndKam-CoNEXT-14},
show that the load factor can be made arbitrarily close to one while keeping the failure rate low, by a large enough choice of the block size~$b$. With this choice and a small false positive rate $\epsilon$, the storage cost is $\bigl(1+o(1)\bigr)\log(1/\epsilon)$ bits per element, matching the information-theoretic lower bound on any approximate set data structure. However, this combination of low storage costs and low failure rate has been observed only in experiments.
The only theoretical analysis so far, also by Fan et al.~\cite{FanAndKam-CoNEXT-14}, is a lower bound showing that $f$ must be $\Omega((\log n)/b)$  to have a low failure probability. They did not provide any matching upper bound showing that some combination of parameters can ensure a low failure probability.

In this paper we provide the first theoretical guarantees on the performance of cuckoo filters. To do so, we describe a simplified version of cuckoo filters, in which we determine the two cells for an element by using its fingerprint directly rather than by using a hash of its fingerprint. This simplification had previously been considered by Fan et al.~\cite{FanAndKam-CoNEXT-14}, but they discarded it without publishing any experimental test results for it. Like Fan et al., we do not expect this simplification to be a practical improvement, but it makes the data structure more amenable to analysis. Under this simplification, we show that the cuckoo filter has (up to lower-order terms, for the same choice of block size, with polynomially small failure probability) nearly the same maximum load rate as the blocked cuckoo hash table of Dietzfelbinger and Weidling~\cite{DieWei-TCS-07}, as long as a constraint that $f=\Omega((\log n)/b)$ is also satisfied. In particular, this analysis allows for the load factor to be arbitrarily close to one, controlled by the block size~$b$. Thus, $f=\Omega((\log n)/b)$ is both a necessary and a sufficient condition for the high-probability success of cuckoo filters.
We also describe how to add a stash to the simplified cuckoo filter, allowing the cuckoo filter to take advantage of the improved reliability of cuckoo filtering with a stash~\cite{KirMitWie-SICOMP-10} without any change in the false positive rate.

Our analysis uses the unrealistic assumption of a uniformly random hash function. However, it uses only two properties of this function: that blocked cuckoo hashing using it succeeds with high probability, and that with high probability it balances the load of a set of balls distributed into a significantly smaller number of bins. Therefore, it is likely that, if the analysis of cuckoo hashing with realistic hash functions~\cite{PatTho-JACM-12,AumDieWoe-Algo-14} is extended to blocked cuckoo hashing, the same analysis can also be extended to cuckoo filters.
It also seems likely that the original version of cuckoo filters behaves at least as well as the simplified version, but we leave the problem of proving this as open for future research. Our algorithm for cuckoo filtering with a stash depends in an essential way on the structure of the simplified cuckoo filter, so extending it to the original cuckoo filter also remains open.

\section{Preliminaries}

We begin by briefly reviewing the Bloom filter, whose operations the cuckoo filter emulates, and the cuckoo hash table on which the organization of a cuckoo filter is based. We then describe the cuckoo filter itself, in the original version given by Fan et al.

As a notational convenience, we use $\log$ without a base to refer to the binary logarithm $\log_2$. We will also use the natural logarithm, denoted by $\ln$.

\subsection{Bloom filter}

A Bloom filter represents a set of $n$ elements by an array of $N$ cells ($N>n$), each containing a single bit of information, together with a hash function mapping the potential elements of the set to $k$-tuples of cells (for a chosen constant parameter value $k$).  A cell contains a nonzero bit if at least one of the elements is mapped to it by the hash function, and a zero bit otherwise. To insert an element into the set, the hash function is used to find its cells, and all of these cells are set to nonzero. To query whether an element belongs to a set, all of its cells are examined, and the result of the query is positive if and only if they are all nonzero. There is no deletion operation.

A false positive occurs if an element that is not part of the given set coincidentally has all of its cells nonzero. For a given choice of $N$, and a given set size $n$, the optimal false positive rate is achieved by setting $k\approx \frac{N \ln 2}{n}$, so that with high probability approximately half of the cells in the table are nonzero.
With these choices, the false positive rate is approximately $2^{-k}$.
Inverting this calculation, the Bloom filter data structure achieves a false positive rate of $\epsilon$
using approximately $\frac{1}{\ln 2}\log(1/\epsilon)\approx 1.44\log(1/\epsilon)$
bits of storage per element~\cite{FanAndKam-CoNEXT-14}.

Many extensions of Bloom filters have been studied. For instance, a \emph{counting Bloom filter}~\cite{FanCaoAlm-TN-00} stores a counter instead of a bit per cell; it can handle deletions, and can also be used (with a smaller number of cells) as a \emph{count-min sketch} to estimate the frequency of items in a data stream~\cite{CorMut-JAlg-05}. An \emph{invertible Bloom filter} adds even more information per cell in order to be able to recover the identities of the set elements stored in it, when there are few enough elements; it also allows deletions, and can be used to find stragglers in a data stream~\cite{EppGoo-TKDE-11}, or as a sketch to communicate the symmetric difference of two similar sets using an amount of communication proportional to the difference~\cite{EppGooUye-SIGCOMM-11}. However, these methods blow up the size of the data structure by a nonconstant factor, and so are less suitable for the original task of the Bloom filter, of representing approximate sets using very little memory.

\subsection{Cuckoo hashing}

Cuckoo filters are based on \emph{cuckoo hashing}, one of many hashing based techniques for maintaining a collection of key--value pairs and looking up the value associated with a query key~\cite{PagRod-JAlg-04}. Cuckoo hashing is a form of open addressing, a family of hashing techniques in which each cell of a hash table stores a single key--value pair. In cuckoo hashing, each key has only two locations in which it may be stored, which are determined by a hash function. Thus, answering a query is simple: look in those two cells and test whether either cell contains the query key.

Inserting a key into a cuckoo hash table is more complicated. If one of the two cells for the key is empty, it can be inserted there. But otherwise, one of the other keys occupying one of these two cells must be kicked out, to make way for the new key. The kicked-out key must then be re-inserted into its second location, possibly kicking out another key there, and so on. This process will either eventually terminate with all keys stored in one of their two cells, or it may fail and force the data structure to be rebuilt. A failure may occur, for instance, when some set of $q$ keys is mapped to fewer than $q$ cells, so there is not enough room to store all of these keys in their cells.

In analyzing this structure, we make the standard assumption that the two cells for each key are chosen uniformly at random, independently from each other and from all the other keys. However, there has also been much research on practical hash function algorithms that do not obey this assumption but nevertheless can be made to work with cuckoo hashing~\cite{DieSch-SODA-09,PatTho-JACM-12}.

Two of the shortcomings of this basic version of cuckoo hashing are that the failure probability is only moderately small (proportional to $1/n^2$ per insertion, for a hash table with a constant load factor, rather than being adjustable to arbitrary inverse polynomials) and that the load factor it can tolerate while achieving this failure probability is also bounded away from $1$ (in fact, bounded below $1/2$). Because of these issues, researchers have investigated modifications of cuckoo hashing that can tolerate higher loads with improved failure probabilities. For the results that we report on in this paper, we need to understand two such modifications, blocked cuckoo hashing~\cite{DieWei-TCS-07} and cuckoo hashing with a stash~\cite{KirMitWie-SICOMP-10}.

\subsection{Blocked cuckoo hashing}

Blocked cuckoo hashing was initially developed by Dietzfelbinger and Weidling~\cite{DieWei-TCS-07}; we follow here its description by Kirsch et al.~\cite{KirMitWie-SICOMP-10}. In blocked cuckoo hashing, each cell of the hash table stores a block of up to $b$ different key--value pairs, for a parameter $b$ chosen as part of the implementation or initialization of the data structure. A query may examine all of the pairs in the two cells that it searches; however, the locality of reference of the query is still as good as in the original version of cuckoo hashing.

When a key is inserted, and one of its two cells is not full (has fewer than $b$ keys already stored in it) it may be placed directly in that cell. However, when both of its cells are full, one of the keys already placed in one of those cells must be kicked out, and moved to its other location. As with standard cuckoo hashing, this move may cause another key to move, possibly creating a chain of dislocations. This sequence of moves can alternatively be viewed as an augmenting path in a graph whose vertices are table cells and whose edges are the pairs of cells that each key maps to. Kirsch et al. write that, for the analysis of the failure probability of this algorithm, it is unimportant how the augmenting path is found, but that an analysis of Dietzfelbinger and Weidling~\cite{DieWei-TCS-07} shows constant expected time (in the event of no failures) for a breadth-first algorithm for finding these augmenting paths.

In order to achieve a load factor of $1-\delta$, blocked cuckoo hashing may be used with any block size $b\ge 1+\frac{\ln(1/\delta)}{1-\ln 2}$. Thus, for constant $\delta$, the block size is also a constant. With this block size and load factor, the failure probability per insertion is $O(1/n^b)$~\cite{KirMitWie-SICOMP-10}.

\subsection{Cuckoo hashing with a stash}

A \emph{stash} is a small collection of key--value pairs that have not been included in a cuckoo hash table.  In cuckoo hashing with a stash, the stash is used to store key--value pairs whose insertion would otherwise cause the hash table to fail. To perform a query in a cuckoo hash table with a stash (in either the original or blocked form of cuckoo hashing) one first checks the two cells that can contain the query key. Then, if the key is not found in either of those two locations, and both locations are full, the stash is also searched. This causes an additional sequence of memory accesses for unsuccessful searches (or for searches of keys already in the stash), but does not slow down most searches, and in many cases greatly improves the reliability of this structure~\cite{KirMitWie-SICOMP-10}.

To apply this technique to cuckoo filtering, we need an analysis of blocked filtering with a stash. Kirsch et al.~\cite{KirMitWie-SICOMP-10} claimed that, for blocked cuckoo hashing with block size $b$ and a stash that can hold $\sigma$ key--value pairs, the failure probability per insertion is $O(1/n^{(\sigma+1)(b-1)+1})$. Unfortunately, this multiplicative improvement in the exponent of failure probability, compared to the version without the stash, is incorrect. As Martin Dietzfelbinger and Michael Rink observed, a failure mode in which $\sigma+2b+1$ keys all map to the same pair of cells already causes the failure probability to be much larger than this bound.\footnote{Michael Mitzenmacher, personal communication, February 7, 2016.} Therefore, any improvement to the reliability of blocked cuckoo hashing, obtained by adding a stash, should be considered conjectural.

\subsection{Cuckoo filter}

A \emph{cuckoo filter}~\cite{FanAndKam-CoNEXT-14} modifies the blocked cuckoo hash table by storing a small \emph{fingerprint} for each element in a set, instead of storing a key--value pair. As in the blocked cuckoo hash table, each cell of the table can store a small number $b$ of (fingerprints of) elements. For each element there are two cells at which its fingerprint could be stored.
To test whether an element belongs to the filter, we examine these two cells and report yes when a matching fingerprint is found in one of them.

In order to pack its fingerprints into these cells, the cuckoo filter (like a cuckoo hash table) may sometimes move them to the other location for their key. However, when it does this, it will not know the  key from which the fingerprint was generated. Therefore, it must determine the other location for a fingerprint from the fingerprint alone. This limitation means that the two locations for the fingerprint of a key can no longer be chosen independently of each other, complicating the analysis of this data structure.

To keep things simple, we will assume that the number of cells $N$ in the cuckoo filter is a power of two.\footnote{It is tempting to try modular addition in place of exclusive ors to extend this method to other choices of $N$, but this fails because exclusive or is an involution and modular addition isn't. However, the simplified cuckoo filter that we describe later can have a number of subtables that is not a power of two.} Thus, if $x$ and $y$ are indices into the table (numbers in the range from $0$ to $N-1$ inclusive), we can combine them by a bitwise exclusive or operation giving a number $x\oplus y$ that is also an index into the table.
Cuckoo filtering depends on three hash functions, $\fp$, $h_1$, and $h_2$, assumed (for purposes of analysis) to be independent random functions:
\begin{itemize}
\item Function $\fp$ maps each potential set element to its fingerprint, an $f$-bit binary number. It is convenient to restrict $\fp$ to have nonzero values so that a zero fingerprint can be used to mark an unused cell of the cuckoo filter; this has no significant effect on the asymptotic behavior of the structure.
\item Function $h_1$ maps each potential set element to a number from $0$ to $N-1$. This gives the location of one of the two hash table cells into which the fingerprint for that element can be placed.
\item Function $h_2$ maps fingerprints to numbers in the range from $1$ to $N-1$ (inclusive). This is not the index of a hash table cell, but rather the difference (or more precisely the bitwise exclusive or) of any pair of locations at which that fingerprint should be stored.
\end{itemize}
Thus, using these functions, the fingerprint $\fp(x)$ for any element $x$ of the given set will be stored either in cell $h_1(x)$ or cell $h_1(x)\oplus h_2(\fp(x))$ of the hash table. If some subset of two or more elements all have the same fingerprint and are all mapped by $h_1$ to the same cell, then each of the fingerprints for these elements will be stored separately in one of the two cells for these elements, even when that would cause two copies of the same fingerprint to be stored in the same cell. That is, in order to make deletions possible, we do not allow different elements to share a copy of a stored fingerprint.

To query whether a value $x$ is a member of the set represented by a cuckoo filter, we examine the $2b$ fingerprints stored at cells $h_1(x)$ and $h_1(x)\oplus h_2(\fp(x))$, and check whether any of them equal $\phi(x)$; if so we report that $x$ is indeed a member. Thus, this check always answers correctly when $x$ belongs to the set, and may give a false positive with probability at most $2b/2^{f-1}$ when $x$ does not belong to the set. This query may be performed using a number of binary word operations (including multiplication of binary numbers) proportional to the number of words needed to store the two cells, rather than being proportional to the larger number of fingerprints stored in these cells; see \autoref{sec:bit-parallel} for details.

To insert a value $x$ into the set, we compute its fingerprint and place this fingerprint into one of the two cells associated with $x$, possibly relocating other fingerprints along an augmenting path in the graph of cells and pairs of cells associated with the elements of the set. The second location for each fingerprint can be calculated from its first location and the fingerprint itself, without needing to know the element from which the fingerprint is generated; otherwise, the insertion operation proceeds exactly as in the blocked cuckoo hash table. Following Kirsch et al.~\cite{KirMitWie-SICOMP-10}, we do not specify precisely how the augmenting path is to be found, except to note that one possible choice for finding it is to use a breadth first search.

To remove a value $x$ from the set, we find a matching fingerprint for $x$, and remove it from its cell. If there are multiple matching fingerprints, we remove only one copy; it does not matter whether that copy is the one created when $x$ was inserted.

\section{The simplification and its graph}

In the simplified version of cuckoo filters that we analyze here, we omit the second hash function, $h_2$. Thus, given any set element~$x$, the two locations in which we may store the fingerprint $\fp(x)$ are the cells with index $h_1(x)$ and $h_1(x)\oplus\fp(x)$. Other than this change, the operation of the simplified cuckoo filter remains the same as in the original version.

\begin{figure}[t]
\centering
\includegraphics[width=0.75\textwidth]{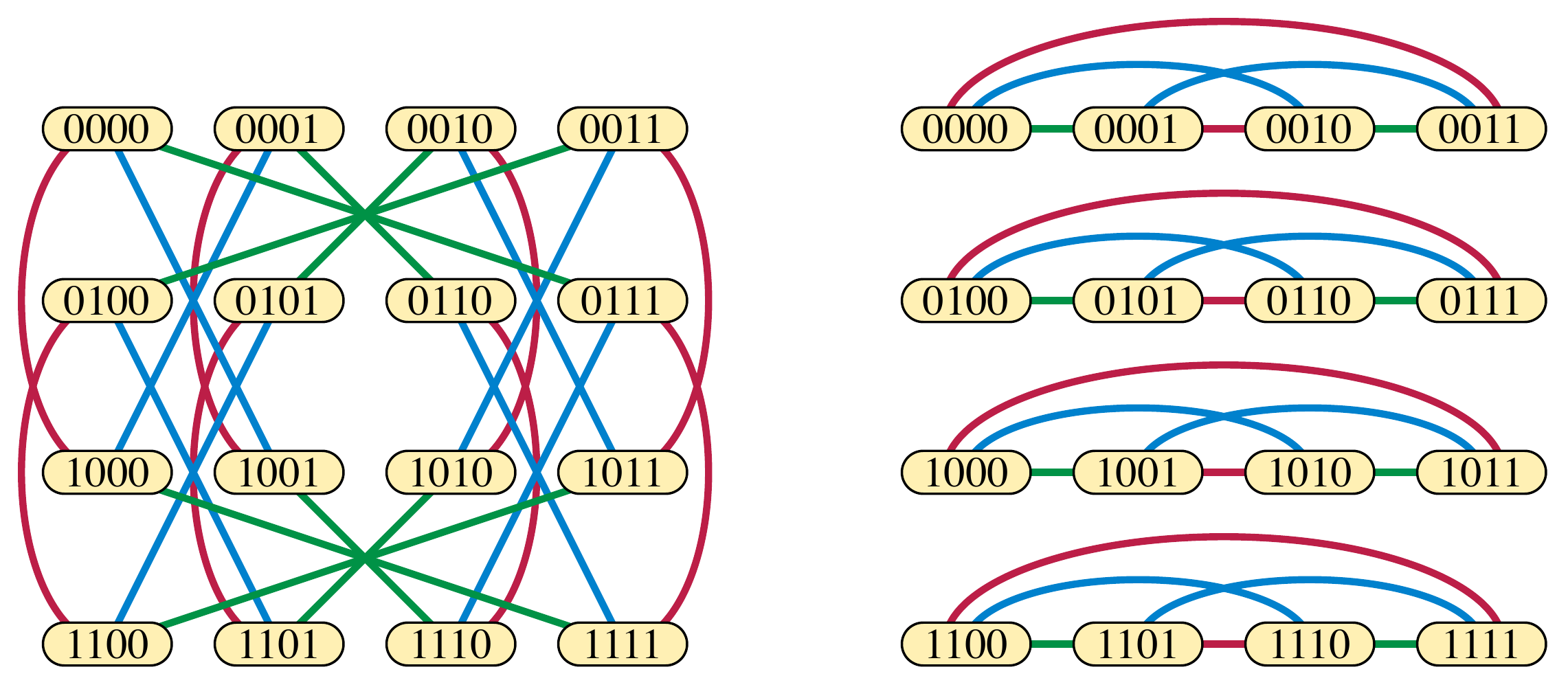}
\caption{The graph of cuckoo filter cells, and pairs of cells that can be the two locations of any fingerprint, for a cuckoo filter with 16 cells and two-bit fingerprints. The edges for each fingerprint are colored green for fingerprint $01$, blue for fingerprint $10$, and red for fingerprint $11$. Left: the original cuckoo filter, with $h_2(01)=0111$, $h_2(10)=1001$, and $h_2(11)=1000$. Right: the simplified cuckoo filter without $h_2$.}
\label{fig:filter-graph}
\end{figure}

The effect of this simplification may be visualized using a graph whose vertices are the cells of the hash table, and whose edges connect the possible pairs of locations of a single fingerprint (\autoref{fig:filter-graph}). In the original cuckoo filter, the endpoints of each edge differ (in their exclusive or) by one of the values of the hash function $h_2$. The resulting graph is regular, as each vertex has one incident edge for each possible hash function value, and has the symmetries of a hypercube. When $2^f>\log n$, the graph has  probability $\Omega(1)$ of being connected, in which case it is a Cayley graph of the group generated by exclusive ors of the values of $h_2$. For smaller values of $f$, with probability $\Omega(1)$ the hashes of all the fingerprints will be independent vectors over the two-element field, and if so the graph will be a disjoint union of $N/2^{2^f-1}$ hypercube subgraphs with $2^{2^f-1}$ vertices per hypercube.
However, when the second hash function $h_2$ is omitted, the corresponding graph of cell pairs is much less well-connected: it is a disjoint union of $N/2^f$ subgraphs, each of which is a clique. The graph of the simplified cuckoo filter is not connected unless $f>\log n$, a much stronger requirement than is needed for the likely connectivity of the graph of the original cuckoo filter.

This simplification had previously been considered by Fan et al.~\cite{FanAndKam-CoNEXT-14}. However, they discarded it without publishing any test results for it.
They write:
\begin{quotation}
``If the alternate location were calculated \dots without hashing the fingerprint, the items kicked out from nearby buckets would land close to each other in the table \dots{} Hashing the fingerprints ensures that these items can be relocated to buckets in an entirely different part of the hash table, hence reducing hash collisions and improving the table utilization.''
\end{quotation}

Although we agree with this criticism, the simplification has the advantage that it makes the method more amenable to analysis. In particular, we can bound the failure probability of the simplified cuckoo filter by treating each connected component of the graph described above as an independent cuckoo hash table. Within each component, all pairs of table cells are equally likely to be chosen by any element that maps to that component, so the previous analysis of cuckoo hash tables may be applied directly.  Although we expect the failure probability for the original cuckoo filter to be at least as good as for the simplified version, we do not know how to extend our analysis to it.

We remark that, instead of randomly choosing a function $h_2$, and constructing a graph of cells and fingerprint-edges based on it, it would be possible to develop a generalized version of cuckoo filtering whose graph is any desired $(2^f-1)$-regular graph on the cells of the cuckoo filter. If this graph is also $(2^f-1)$-edge-colorable, this may be done by associating each fingerprint with an edge color, and making the two cells that can store the fingerprint of an element $x$ be $h_1(x)$ and the neighbor of $h_1(x)$ along the edge with color $\fp(x)$. In the more general case, in which the graph is not $(2^f-1)$-edge-colorable, this may be done by choosing a one-to-one correspondence between fingerprints and outgoing locations at each cell, and either storing $\fp(x)$ in $h_1(x)$ or storing $\varphi(x)$ in the neighbor of $h_1(x)$ associated with fingerprint $\fp(x)$, where $\varphi(x)$ is the fingerprint associated with the edge back to $h_1(x)$. In this way, for instance, it would be possible to make the graph connected even when $2^f\le\log n$, something that is not possible for the original cuckoo filter.

\section{Analysis}

We have seen by the graphical analysis in the previous section that, in the simplified version of cuckoo filtering that we study here, the table of cells can be partitioned into connected components of the fingerprint-edge graph, each of which is a clique. A single connected component consists of $2^f$ cells. Each two cells within a connected component have locations that differ only within their least significant $f$ bits. We call each of these connected components a \emph{subtable}. Essentially, each subtable is itself a cuckoo filter, on a subset of the input data. We analyze two different failure modes of the global filter: either it can fail to uniformly distribute the set elements to its subtables, or the cuckoo filtering within a subtable can fail.

\subsection{Even distribution into subtables}

We assume a cuckoo filter representing $n$ elements, with $b$ fingerprints per cell and
$N=(n/b)/(1- \delta-\delta^2)$ cells where $\delta$ and $b$ are related as $b\ge 1+\frac{\ln(1/\delta)}{1-\ln 2}$, the same relation used by Kirsch et al. (Prop.~4.1) for blocked cuckoo hashing~\cite{KirMitWie-SICOMP-10}. We also assume that the fingerprints of the cuckoo filter have $f$ bits each, so that the cuckoo filter may be partitioned into $N/2^f$ subtables of $2^f$ cells per subtable.

The subtable into which an element $x$ is mapped is given by the most significant $-f+\log N$ bits of $h_1(x)$.  Assuming that $h_1$ is a random hash function, the probability that each element falls into a particular subtable $S$ is $2^f(1/N)$, and the events that elements fall into a subtable (for different elements) are mutually independent. Thus, the number of elements that are mapped into subtable $S$ is the sum of $n$ i.i.d. Bernoulli random variables with probability $p=2^f(1/N)$ of being $1$ and probability $1-p$ of being $0$. The expected number $\mu$ of elements that are mapped into subtable $S$ is $2^f (n/N) = 2^f b(1- \delta-\delta^2)$.

We will say that a subtable is \emph{overfull} when the number of elements mapped to it exceeds $\mu(1-\delta)/(1-\delta-\delta^2)=\mu(1+\delta^2+o(\delta^2))$, giving it a load factor greater than $1-\delta$. By a standard form of the multiplicative Chernoff bound, for $\delta$ sufficiently smaller than 1, the probability that $S$ is overfull is at most $\exp\bigl(-\delta^4\mu/3\bigr)$.
To achieve probability $1/n^s$ of avoiding any overfull subtable, it suffices (by the union bound) to achieve probability $1/n^{s+1}$ that one table $S$ is overfull.
Plugging $\mu\ge 2^f b$ into the probability that $S$ is overfull and solving for $f$ gives that no table is overfull, with high probability $\ge 1-1/n^s$, whenever
\[
f\ge \log\left(\frac{3(s+1)\ln n}{\delta^4 b}\right)
=\log\left(\frac{\log n}{b}\right)+O(1).
\]
where the simplification on the right hand side is based on the assumption that $\delta$ and $s$ are constants. This $\Omega\bigl(\log((\log n)/b)\bigr)$ constraint on~$f$ will be insignificant in comparison with the $\Omega\bigl((\log n)/b\bigr)$ constraint coming from the failure probability within each subtable.

We summarize the results of this subsection as a lemma:

\begin{lemma}
Let $n$, $N$, $b$, $f$, and $\delta$ be as above.
Suppose also that $f=\log\bigl((\log n)/b\bigr)+\Omega(1)$.
Then the probability that a simplified cuckoo filter with these parameters has any overfull subtables
is polynomially small, with an exponent that can be made arbitrarily large by using a larger constant factor in the $\Omega$-notation of the bound for~$f$.
\end{lemma}

\subsection{Failure probability within each subtable}

Let $S$ be a subtable that is not overfull; that is, at most $2^fb/(1+\delta)$ elements of the given $n$-element set are mapped into it, where again $f$ is the fingerprint length, $b$ is the number of fingerprints per cell, and $\delta$ is a constant that is yet to be determined. If an element $x$ is mapped to subtable $S$, then the location within $S$ to which it is mapped (the low-order bits of $h_1(x)$) are independent of the information causing it to be mapped to $S$ (the high-order bits of $h_1(x)$), so these locations are uniformly distributed within $S$. Additionally, the fingerprint $\phi(x)$ is uniformly distributed among all of the valid fingerprints, and the set of locations given by the bitwise exclusive or of these fingerprints with $h_1(x)$ is exactly the set of all remaining locations within $S$. Therefore, we may analyze each subtable independently, as a data structure containing at most $2^fb/(1+\delta)$ fingerprints, each of which is mapped to a uniformly random pair of distinct locations within the subtable. That is, limiting our attention to a single subtable has eliminated the dependence between the pairs of locations used by each element of the cuckoo filter.

Such a data structure behaves exactly the same as a blocked cuckoo hash table with the same elements. Thus, we can apply the previously known analysis of a blocked cuckoo hash table directly. However we must keep in mind the fact that, because we are studying events that hold with high (inverse-polynomial) probability, the reduced size of a subtable (relative to an $n$-element cuckoo hash table) translates into weaker bounds on the failure probability.

Choose $b$ and $\delta$ so that $b\ge 1+\frac{\ln(1/\delta)}{1-\ln 2}$, the same constraints on $b$ and $\delta$ made in the analysis of blocked cuckoo hash tables. With these choices, $\delta$ may be made arbitrarily small by choosing $b$ large enough, and in particular the factor $(1+\delta)^2$ appearing in the analysis of how evenly the subtables are distributed may also be made arbitrarily close to one. Then by the previous analysis of blocked cuckoo hash tables, the failure probability per insertion, for a subtable containing $\Theta(b2^f)$ elements, is $O\bigl((b2^f)^{-b}\bigr)$, inversely proportional to the $b$th power of the number of elements.

We desire the data structure to succeed with high probability; that is, for an arbitrarily chosen constant $s$, we should be able to achieve failure probability $O(1/n^s)$. This will be true when $(b2^f)^b=\Omega(n^s)$ or, taking logarithms of both sides, when
$b(f+\log b)>s\log n+O(1)$. The $b\log b$ term can be assumed to cancel the $O(1)$, giving
$f>(s\log n)/b$ as a sufficient condition for high-probability success. This matches the $f=\Omega\bigl((s\log n)/b\bigr)$ necessary condition of Fan et al.~\cite{FanAndKam-CoNEXT-14}
(based on a calculation of the probability that more than $2b$ elements hash to the same location/fingerprint pair), which also applies to this version of cuckoo filters.

We summarize the results of this subsection as a lemma:

\begin{lemma}
\label{lem:subtable-failure}
Let $b$ be any sufficiently large constant, and let $s$ also be constant. Let $n$, $N$, $\delta$, and $f$ be as defined above, and suppose that $f>(s\log n)/b$. Then the probability of failure for an insertion of an element into a simplified cuckoo filter with these parameters and without any overfull subtables  is at most $1/n^s$.
\end{lemma}

\subsection{Overall failure probability}

Combining our results on the two failure modes of cuckoo filters, we have the following result.

\begin{theorem}
Let $n$ and $b$ be given, and let the maximum load factor for high-probability-of-success blocked cuckoo hashing with block size $b$ be $1-\delta$, where $b\ge 1+\frac{\ln(1/\delta)}{1-\ln 2}$.
Then creating a cuckoo filter for $n$ elements with block size $b$ and fingerprint size $f$ succeeds with high probability for load factor $1-\delta-\delta^2$ when $f=\Omega\bigl((\log n)/b\bigr)$. More specifically, to achieve probability $O(1/n^s)$ of failure it is sufficient for $b$ to be a sufficiently large constant and for $f$ to obey the inequality $f>(s+1)(\log n)/b$.
\end{theorem}

The $1+o(1)$ factor by which the load factor $1-\delta$ of blocked cuckoo hashing differs from the load factor $1-\delta-\delta^2$ of this result is caused by the fact that, when $n$ elements are filtered into subtables, some subtables will likely be larger than their expected size. The term $\delta^2$ could be replaced here by any fixed power of $\delta$.
One way of avoiding this change in load factor altogether would be to redo the analysis of blocked cuckoo hashing based on the assumption that the input is chosen by including elements independently at random at a fixed rate, rather than that the set of input elements has a fixed size. With this modified input model, the same model with the same rate would automatically apply to each subtable, without need of Chernoff bounds. However, this refinement would make little difference to our overall results.

\section{Cuckoo filter with a stash}

The extra structure of the subtables in the simplified cuckoo filter that we analyze here also makes it easier to add a stash to this structure, to amplify its success probability.
Specifically, we add a separate stash for each subtable of the filter.
Each stash will store a collection of (location,fingerprint) pairs, for fingerprints that were not able to be stored within its subtable. Because the stash is specific to the subtable, the location part of the stash need only store the low-order $f$ bits of the location, specifying one of the two cells within the subtable that its fingerprint could have been stored in.

Recall that, in the simplified cuckoo filter,
Each fingerprint $\fp(x)$ can be stored in two locations, $h_1(x)$ and $h_1(x)\oplus\fp(x)$. However, after storing the fingerprint in the filter, we no longer know which of its two locations is $h_1$.
Therefore, in the pair that we store in the stash, we use the location with the smaller index, given by (the low-order $f$ bits of) $\min(h_1(x),h_1(x)\oplus\fp(x))$.

To query whether an element $x$ belongs to the set represented by a cuckoo filter with a stash, we test whether either of the two locations for $x$ in the filter contains the fingerprint for $x$. If not, we determine the minimum of those two locations, and search for the resulting (location,fingerprint) pair in the stash for the subtable of $x$. To insert an element $x$ into the filter, we attempt to insert it into the cuckoo filter as before, and if this fails we add it to the stash. And to delete an element $x$, we perform a query to determine which of the two locations for $x$ contains its fingerprint, or whether the fingerprint is located in the stash, and remove one copy of this fingerprint from one of its locations.

The same analysis as for the cuckoo filter without a stash goes through in the same way as before, replacing the failure probability of a blocked cuckoo hash table (used to prove \autoref{lem:subtable-failure}) with any improved failure probability (currently only conjectural) that can be proved for a blocked cuckoo hash table with a stash.

The failure probability can also be boosted, with theoretical guarantees rather than conjectures on the improvement, by increasing the block size~$b$ rather than adding a stash. Adding a stash to each subtable would increase the memory requirements of the data structure by only a lower-order term, but increasing the block size could actually decrease the memory requirements (assuming the fingerprint size $f$ is kept constant) by allowing a load factor closer to~$1$ to be used. So, given that stashes are not an improvement over increased block sizes in the guarantees they provide, in the reliability that can be obtained with them, or in the storage space they use, why would one ever use a stash? The answer lies in another feature of their analysis, their effect on the false positive rate of the data structure.

Recall that, in a cuckoo filter with block size $b$ and fingerprint size $f$, the false positive rate is at most $2b/(2^f-1)$: each query examines at most $2b$ fingerprints, each of which has a $1/(2^f-1)$ chance of colliding with a given query of an element that does not belong to the set.
The probability is at most $2b/(2^f-1)$ rather than exactly $2b/(2^f-1)$ for two reasons: a block of cells may not be full, or it may have more than one copy of the same fingerprint. However, for typical choices of parameters, neither of these reasons gives a large effect on the false positive rate.
Based on this calculation, increasing the block size from $b$ to $b'$ would also increase the false positive rate by the same $b'/b$ factor. Naively, a stash of size $\sigma$ would again increase the false positive rate to $(2b+\sigma)/(2^f-1)$, since now there are potentially $\sigma$ additional fingerprints that could collide with any given query. However, as we now prove, the stash does not actually change the false positive rate at all. The analysis below also gives a tighter analysis on the false positive rate for cuckoo hashes even without a stash, taking into account the possibility of under-full blocks.

\begin{theorem}
The false positive rate of a cuckoo filter with $b$ blocks, $n$ elements, $N$ hash table cells, fingerprint size $f$, and load factor $1-\delta=n/Nb$, regardless of whether it uses a stash or not, is at most $\tfrac{2n}{N(2^f-1)} = \tfrac{2b(1-\delta)}{2^f-1}$.
\end{theorem}

\begin{proof}
There are $n$ elements in the set represented by the filter, each of which independently selects a hash table location and fingerprint. A query collides with an element if and only if it has the same fingerprint and is mapped to one of the same two locations. Based on this choice, any single element has a $2/N(2^f-1)$ probability of colliding with the given query; the factor of $2$ in the numerator comes from the fact that two different choices of hash table location give rise to the same pair of cells in which the fingerprint can be stored. The result follows by the union bound.
\end{proof}

Based on this result, and assuming that adding a stash to a blocked cuckoo hash table can be proved to improve its reliability, the same improvement to reliability can be obtained in a cuckoo filter with a stash, without sacrificing the failure rate or compensating for it by increasing the fingerprint size.

\section{Bit-parallel querying}
\label{sec:bit-parallel}

A query in a cuckoo filter involves testing whether a given fingerprint $\fp(x)$ is stored in one of two cells of the filter. However, a fingerprint may have many fewer bits than a word in the machine architecture on which the filter is implemented; therefore, it may be necessary to test whether the same fingerprint appears at each of several positions in a machine word. We describe here how to perform this task in constant time (independent of the number of positions to be tested) using only bitwise binary operations and arithmetic operations on binary numbers. This analysis shows that cuckoo filter queries can take constant time even for non-constant block sizes, as long as the fingerprint length multiplied by the block size is $O(\log n)$, machine words store $\Omega(\log n)$-bit words, and multiplication takes constant time.

Suppose that each fingerprint has $f$ bits, and that $b$ fingerprints are packed into a word $w$ of $bf$ bits. We perform the following steps:
\begin{itemize}
\item Let $F=\sum_{i=0}^{b-1}2^{if}$, a word in which $b$ fingerprints are packed, all equal to the number $1$. This step can be done when the filter is created, so its calculation does not figure into the time complexity of the query algorithm.
\item Compute $q=w\oplus (\fp(x)\oplus (2^f-1))F$. The subexpression $\fp(x)\oplus (2^f-1)$ creates a fingerprint complementary to $\fp(x)$, and the subexpression $(\fp(x)\oplus (2^f-1))F$ packs $b$ copies of this complementary fingerprint into a single word. Thus, the whole expression, which combines $w$ with these complementary fingerprints by a bitwise exclusive or, gives a word packed with $b$ fingerprints which are all-ones when they match $\fp(x)$, and which have at least one zero in their binary representations when they do not match.
\item Compute $r=((q+F)\oplus q\oplus F)\mathbin{\&} 2^fF$. Here the $\&$ operator represents bitwise Boolean and. The expression $q+F$ produces a carry above each matching fingerprint, and no carry above the fingerprints that do not match. The expression $q\oplus F$ has, in the positions of each of these carry bits, a binary value representing what would be in that position if no carry occurred. The exclusive or of these two subexpressions gives the carry bits that differ from their non-carry values. Masking with $2^fF$ keeps only the bits in the carry positions of this computation, setting the other bits to zero.
\item If $r$ is zero, there is no match. Otherwise, there is a match at the position given by shifting the least significant nonzero bit of $r$ right by $f$ positions. This least significant nonzero can be calculated as $r\mathbin{\&}\mathop{\sim}(r-1)$, where $\sim$ is the bitwise complement operator.
\end{itemize}
Once a value with a nonzero bit at the position of the first match has been obtained, it is straightforward to perform other operations such as removing the fingerprint at that position; we omit the details.
We summarize the results here as a theorem:

\begin{theorem}
Suppose that a cuckoo filter with block size $b$ packs each block of $b$ fingerprints into $B$ machine words, in a model of computation in which bitwise Boolean operations and binary number arithmetic (including multiplication) take constant time per operation. Then each query in the filter may be performed in time $O(B)$.
\end{theorem}

\bibliographystyle{plainurl}
\bibliography{cuckoo}

\end{document}